\documentclass[9pt,journal]{IEEEtran}

\usepackage{amsmath,amsfonts,amssymb,mathtools}
\usepackage{hhline}
\usepackage{cite}
\usepackage{graphicx}
\usepackage{psfrag}
\usepackage{subfigure}
\usepackage{stfloats}
\usepackage{array}
\DeclareMathAlphabet\mathbfcal{OMS}{cmsy}{b}{n}

\newtheorem{Lm}{Lemma}

\newtheorem{remark}{Remark}
\newtheorem{exmp}{Example}

\newcommand{\tr}[1]{\mathop{{\rm \bf tr}\left[#1\right]}\nolimits}
\newcommand{\E}[1]{\mathop{{\rm \bf E}\left\{#1\right\}}\nolimits}

\newcommand{\D}[1]{\mathop{{\rm d}#1}}                  % bez scobok
\newcommand{\DD}[1]{\mathop{\left( {\rm d}#1 \right)}}  % so scobkami

\newcommand{\Di}[1]{\mathop{\partial_i #1}}                  % bez scobok
\newcommand{\DDi}[1]{\mathop{\left( \partial_i #1 \right)}}  % so scobkami

\usepackage{clrscode}
\makeatletter
\g@addto@macro\code@init\normalfont
\makeatother

\begin{document}
\title{SVD-based Kalman Filter Derivative Computation}
%\title{Derivation of the Kalman Filter Sensitivity Equations with Use of Singular Value Decomposition}

\author{J.V.~Tsyganova and M.V.~Kulikova
\thanks{Manuscript received ??; revised ??.  The second author thanks the support of Portuguese National Fund ({\it Funda\c{c}\~{a}o para a
Ci\^{e}ncia e a Tecnologia})  within the scope of project UID/Multi/04621/2013.}
\thanks{The first author is with Ulyanovsk State University,  Str. L. Tolstoy 42, 432017 Ulyanovsk, Russian Federation. The second author is with
CEMAT (Center for Computational and Stochastic Mathematics), Instituto Superior T\'{e}cnico, Universidade de Lisboa,
          Av. Rovisco Pais 1,  1049-001 LISBOA, Portugal;
          Emails:
TsyganovaJV@gmail.com; maria.kulikova@ist.utl.pt}}

%\markboth{IEEE Transactions on Automatic Control}{}
\markboth{PREPRINT}{}

\maketitle

\begin{abstract}
Recursive adaptive filtering methods are often used for solving the problem of simultaneous state and parameters estimation arising in many areas of research. The gradient-based schemes for adaptive Kalman filtering (KF) require the corresponding filter sensitivity computations. The standard approach is based on the direct differentiation of the KF equations. The shortcoming of this strategy is a numerical instability of the conventional KF (and its derivatives) with respect to roundoff errors. For decades, special attention has been paid in the KF community for designing efficient filter implementations that improve robustness of the estimator against roundoff. The most popular and beneficial techniques are found in the class of square-root (SR) or UD factorization-based methods. They imply the Cholesky decomposition of the corresponding error covariance matrix. Another important matrix factorization method is the singular value decomposition (SVD) and, hence,  further encouraging KF algorithms might be found under this approach. Meanwhile, the filter sensitivity computation heavily relies on the use of matrix differential calculus. Previous works on the robust KF derivative computation have produced the SR- and UD-based methodologies. Alternatively, in this paper we design the SVD-based approach. The solution is expressed in terms of the SVD-based KF {\it covariance} quantities and their derivatives (with respect to unknown system parameters). The results of numerical experiments illustrate that although the newly-developed SVD-based method is algebraically equivalent to the conventional approach and the previously derived SR- and UD-based strategies, it outperforms the mentioned techniques for estimation accuracy in ill-conditioned situations.
\end{abstract}


\IEEEpeerreviewmaketitle

\section{Introduction}

The problem of filter sensitivities evaluation plays a key role in many areas of research; for instance, in state estimation and parameter identification realm~\cite{mehra1972,Bastin1988}, in the field of optimal input design~\cite{Mehra1974,Gupta1974}, in information theory for computing the Fisher information matrix~\cite{zadrozny1994,klein1995,klein2000} {\it etc}. In this paper we explore linear discrete-time stochastic systems where the associated Kalman filter (KF) is used for estimating the unknown dynamic states. Therefore, the standard approach for computing the filter sensitivities (with respect to unknown system parameters) is a direct differentiation of the KF equations. This conventional methodology is comprehensively studied in~\cite{Mehra1974,zadrozny1989,klein2000direct}. The shortcoming of this strategy is a numerical instability of the conventional KF (and its derivatives) with respect to roundoff errors discussed in~\cite{VerhaegenDooren1986,Verhaegen1989}. Due to this fact, special attention has been paid in the KF community for designing robust KF implementation methods. The most popular techniques belong to the class of square-root (SR) or UD factorization-based methods; see~\cite{KaminskiBryson1971,Bierman1977,Sayed1994,ParkKailath1995} and many others. These algorithms imply the Cholesky decomposition and its modification for the corresponding covariance matrix factorization~\cite{Bierman1977,KailathSayed2000,GrewalAndrews2001}. We may note that the Cholesky decomposition exists and is unique when the symmetric matrix to be decomposed is positive definite~\cite{Golub1983}. If it is a positive semi-definite, then the Cholesky decomposition still exists, however, it is not unique~\cite{Higham1990}. Further encouraging KF implementation methods might be found with the use of singular value decomposition (SVD). Some evidences of better estimation quality obtained under the SVD-based approach exist in the field of nonlinear filtering; for instance, see discussion in~\cite{Zhang2008,straka2013,Zhang2014} and others. For linear filtering problem examined in this paper, the first SVD-based KF was, to the best of our knowledge, designed in~\cite{WangSVD1992}. Our recent analysis exposes that the mentioned SVD-based filter can be further improved for enhancing its numerical robustness. This result is comprehensively studied in~\cite{KulikovaOnline}, where some new stable SVD-based KF implementations are designed.

Despite the existence of inherently more stable SR-, UD- and SVD-based KF variants, the problem of robust filter derivative computation is seldom addressed in practice because of its complicated matter. The solution to the mentioned problem heavily relies on the use of matrix differential calculus. The first SR-based {\it information-type} algorithm for the KF derivative computations belongs to Bierman {\it et al.} and was appeared in 1990; see~\cite{Bierman1990}. Alternatively, the SR-based {\it covariance-type} method was proposed in~\cite{Kulikova2009IEEE} as well as the UD-based scheme designed in~\cite{Tsyganova2013IEEE}. Later on, a general ``differentiated'' SR-based methodology was designed for both orthogonal and J-orthogonal transformations involved in the filtering equations (and their derivatives) in~\cite{Kulikova2013IEEE,Kulikova2015,Kulikova2016}. Alternatively, in this technical note we develop the SVD-based approach for the KF derivative computation. We show that the new technique is algebraically equivalent to the conventional ``differentiated'' KF, but it improves the robustness against roundoff errors as well as the existing ``differentiated'' SR- and UD-based methodologies. However, motivated by the results obtained in nonlinear filtering realm, we expect that the newly-designed SVD-based method outperforms the previously derived algorithms while solving the parameters estimation problem, especially when the error covariance is ill-conditioned.

\section{Filter sensitivity equations: conventional approach} \label{sec:state}
Consider the state-space equations
  \begin{align}
   x_{k} & =  F(\theta) x_{k-1}\!+B(\theta)u_{k-1}\!+G(\theta) w_{k-1}, \quad k \ge 1,  \label{eq:st:1} \\
    z_k    & =  H(\theta) x_k+v_k,  \;  v_k  \sim {\mathcal N}\left(0,R(\theta)\right), \; w_k \sim {\mathcal N}\left(0,\Omega(\theta)\right)             \label{eq:st:2}
  \end{align}
where $z_k \in \mathbb R^m$, $u_k \in \mathbb R^d$, $x_k \in \mathbb R^n$ and $\theta \in \mathbb R^p$ are,
respectively, the vectors of available measurements, the known deterministic control input, the unknown dynamic state and the unknown system parameters that need to be estimated from the available experimental data, $\{z_1, \ldots, z_N \}$. The process and the measurement noises are independent Gaussian zero-mean white-noise processes that are also independent from the initial state $x_0 \sim {\mathcal N}\left(\bar x_0,\Pi_0(\theta)\right)$. The covariances are assumed to be $\Omega(\theta) \ge 0$, $R(\theta) > 0$ and $\Pi_0(\theta) \ge 0$.
%We assume that the entries of the system matrices  $F_k(\theta) \in \mathbb R^{n\times n}$, $G_k(\theta) \in \mathbb R^{n\times q}$, $H_k(\theta) \in \mathbb %R^{m\times n}$, $Q_k(\theta) \in \mathbb R^{q\times q}$, $R_k(\theta) \in \mathbb R^{m\times m}$ are differentiable functions of $\theta$.

Equations~\eqref{eq:st:1}, \eqref{eq:st:2} represent a set of the state-space models (SSMs). Each of them corresponds to a particular system parameter value. This means that for any fixed value of $\theta$, say $\hat \theta^*$, the system matrices are known, i.e. there is no uncertainty in model~\eqref{eq:st:1}, \eqref{eq:st:2}. For simplicity, throughout the paper we write $F$ {\it etc.} instead of $F(\hat \theta^*)$ {\it etc.} when evaluating at the fixed point $\hat \theta^*$. The associated KF yields the linear minimum least-square estimate of the unknown dynamic state that can be recursively computed via the equations~\cite[Theorem~9.2.1]{KailathSayed2000}:
\begin{align}
e_k & =  z_k-H\hat x_{k|k-1}, \qquad \hat x_{1|0}  =  \bar x_0,  \quad  k  \ge 1,  \label{kf:1} \\
K_{p,k} & = FP_{k|k-1}H^TR_{e,k}^{-1}, \quad R_{e,k}  = R+HP_{k|k-1}H^T, \label{kf:2} \\
\hat x_{k+1|k} & =   F \hat x_{k|k-1}+Bu_k+K_{p,k}e_k    \label{kf:3}
\end{align}
where $\{ e_k \}$ are innovations of the discrete-time KF. The important property of the KF for Gaussian SSMs is $e_k \sim {\cal N}\left(0,R_{e,k}\right)$. The $P_{k|k-1}=\E{ (x_{k}-\hat x_{k|k-1})(x_{k}-\hat x_{k|k-1})^{T}}$ is the one-step ahead predicted error covariance matrix computed as follows:
\begin{equation}
P_{k+1|k} = FP_{k|k-1}F^T\!+G\Omega G^T \! - K_{p,k}R_{e,k}K_{p,k}^T, \; P_{1|0} \! = \Pi_0. \label{kf:4}
\end{equation}

The conventional approach for deriving the related sensitivity model is based on differentiation of the corresponding filtering equations. Let $A(\theta) \in \mathbb R^{m\times n}$, $B(\theta) \in \mathbb R^{n\times q}$ be matrices, which entries are differentiable functions of the parameter vector $\theta \in \mathbb R^{p}$. The $m \times n$ matrix
$\Di{A} = \partial A/ \partial \theta_i$ implies the partial derivative of the $A$ with respect to the $i$-th component of $\theta$, $i=1, \ldots p$.  The $m \times n$ matrix $\D{A} = \sum_{i=1}^{p} \DDi{A} \cdot \DD{\theta_i}$ is the differential form of first-order derivatives of $A(\theta)$.  Taking into account the matrix product rule of differentiation~\cite[p. 955]{neudecker1969}: $\D{\left(AB\right)} = \DD{A}B + A\DD{B}$, and the fact $\D I = 0$, we derive $\D{\left(A^{-1}\right)} = -A^{-1} \DD{A} A^{-1}$ for any square and invertible matrix $A$ (it is also known as the Jacobi's formula);
 see also~\cite[p. 546]{zadrozny1989}. Using these differentiation rules, the necessary differentials of~\eqref{kf:1}-\eqref{kf:4}
can be written as follows~\cite{zadrozny1989,klein2000direct}:
\begin{align}
\D{e_k} & =  -\left[ \DD{H}\hat x_{k|k-1} + H\DD{\hat x_{k|k-1}}\right],    \label{diff:kf:1} \\
\D{\hat x_{k+1|k}} & =   \DD{F} \hat x_{k|k-1}+ F \DD{\hat x_{k|k-1}} + \DD{B} u_{k} \nonumber \\
  & + \DD{K_{p,k}}e_k  + K_{p,k}\DD{e_k},   \label{diff:kf:2} \\
\D{K_{p,k}} & = \DD{F}P_{k|k-1}H^TR_{e,k}^{-1}+F\DD{P_{k|k-1}}H^TR_{e,k}^{-1} \nonumber \\
& + FP_{k|k-1}\DD{H^T}R_{e,k}^{-1} \nonumber \\
& - FP_{k|k-1}H^TR_{e,k}^{-1}\DD{R_{e,k}}R_{e,k}^{-1}, \label{diff:kf:3a} \\
 \D{R_{e,k}} & = \D{R}+\DD{H}P_{k|k-1}H^T+H\DD{P_{k|k-1}}H^T \nonumber \\
& +HP_{k|k-1}\DD{H^T}, \label{diff:kf:3b} \\
\D{P_{k+1|k}}  & = \DD{F}P_{k|k-1}F^T + F\DD{P_{k|k-1}}F^T \nonumber \\
& + FP_{k|k-1}\DD{F^T} +\DD{G}\Omega G^T + G\DD{\Omega}G^T  \nonumber \\
& + G\Omega \DD{G^T} - \DD{K_{p,k}}R_{e,k}K_{p,k}^T \nonumber \\
& - K_{p,k}\DD{R_{e,k}}K_{p,k}^T - K_{p,k}R_{e,k}\DD{K_{p,k}^T}. \label{diff:kf:4}
\end{align}
In deriving the equations above we take into account that $\D{z_k} = 0$ and $\D{u_k} = 0$, because the observations $z_k$ and the control input $u_k$ do not depend on the parameters (i.e. their realizations are independent of variations in $\theta$) and therefore have a differential equal to zero.

We may also note that except for the scalar factor $\D{\theta_i}$, $\Di{A}$ is a special case of $\D{A}$, so that to obtain partial-derivative forms from differential forms, we only have to everywhere replace operator $\D{\left(\cdot\right)}$ with $\Di{\left(\cdot\right)}$ for
$i=1, \ldots p$~\cite[p. 546]{zadrozny1989}.  Hence, from~\eqref{diff:kf:1}~-- \eqref{diff:kf:4} we obtain a set of $p$ vector equations, known as the {\it filter sensitivity equations}, for computing $\Di{\hat x_{k+1|k}}$, $i=1, \ldots p$, and a set of $p$ matrix equations, known as the {\it Riccati-type sensitivity equations}, for computing $\Di{P_{k+1|k}}$, $i=1, \ldots p$. This approach for the KF sensitivity model derivation is called the ``differentiated KF''. Its main drawback is a numerical instability of the conventional KF~\eqref{kf:1}~-- \eqref{kf:4} and inherently its derivative~\eqref{diff:kf:1}~-- \eqref{diff:kf:4} with respect to roundoff errors.

The goal of this paper is to design a robust methodology for updating the ``differentiated'' KF equations above in terms of SVD factors (and their derivatives) of the error covariance matrices $P_{k|k-1}$ instead of using the full matrices $P_{k|k-1}$ (and their derivatives).

\section{SVD factorization-based Kalman filtering \label{filters}}

To the best of our knowledge, the first SVD-based KF was by Wang {\it et al.} and appeared in 1992; see Eqs~(17), (22), (23) in~\cite[pp.~1225-1226]{WangSVD1992}.  Our recent research shows that although that implementation is inherently more stable than the KF~\eqref{kf:1}~-- \eqref{kf:4}, it is still sensitive to roundoff and poorly treats ill-conditioned problems. The cited analysis exposes that the SVD-based filter can be further improved for enhancing its numerical robustness. This result is comprehensively studied in~\cite{KulikovaOnline}, where new stable SVD-based KF implementations are designed. The readers are referred to the cited paper for the detailed derivations, numerical stability discussion and proofs. Here, we briefly outline the principle steps for construction of the most advanced SVD-based KF variant. Next, we extend it to a stable filter sensitivities computation, which is the main purpose of this study.

Consider the SVD factorization~\cite[Theorem~2.8.1]{Tyrtyshnikov2012}: suppose $A \in {\mathbb C}^{m\times n}$, ${\rm rank}\: A = r$. There exist positive numbers $\sigma_1\geq\ldots\sigma_r>0$ and unitary matrices $W \in {\mathbb C}^{m\times m}$ and $V \in {\mathbb C}^{n\times n}$ such that
\[
A  = W\Sigma V^*, \,
\Sigma =
\begin{bmatrix}
S & 0 \\
0 & 0
\end{bmatrix} \in {\mathbb C}^{m\times n},  \; S={\rm diag}\{ \sigma_1,\ldots,\sigma_r\}
\]
where $V^*$ is the conjugate transpose of $V$.

The diagonal entries of $\Sigma$ are known as the singular values of $A$. The non-zero $\sigma_i$ ($i=1, \ldots, r$) are the square roots of the non-zero eigenvalues of both $A^*A$ and $AA^*$.

If $A$ is a square matrix such that $A^*A=AA^*$, then the $A$ can be diagonalized using a basis of eigenvectors according to the spectral theorem, i.e. it can be factorized as follows: $A = QDQ^*$  where $Q$ is a unitary matrix and $D$ is a diagonal matrix, respectively. If $A$ is also positive semi-definite, then the spectral decomposition above, $A = QDQ^*$, is also a SVD factorization, i.e. the diagonal matrix $D$ contains the singular values of $A$. For the SSMs examined in this paper, the initial error covariance $\Pi_0 \in {\mathbb R}^n$ is a symmetric positive semi-definite matrix and, hence, the spectral decomposition implies $\Pi_0 = Q_{\Pi_0}D_{\Pi_0}Q_{\Pi_0}^T$  where $Q_{\Pi_0}$  and $D_{\Pi_0}$ are the orthogonal and diagonal matrices, respectively. It is also a SVD factorization, i.e. the factor $D_{\Pi_0}$ contains the singular values of $\Pi_0$.

Now, we are ready to present the SVD-based KF implementation developed recently in~\cite{KulikovaOnline}. Instead of conventional recursion~\eqref{kf:1}-\eqref{kf:4} for $P_{k|k-1}$, we update only their SVD factors, $\{Q_{P_{k|k-1}}, D^{1/2}_{P_{k|k-1}}\}$, at each iteration step of the filter as shown below.

\textsc{Initial Step}  ($k=0$). Apply the SVD factorization for the initial error covariance matrix $\Pi_0 = Q_{\Pi_0}D_{\Pi_0}Q_{\Pi_0}^T$ and, additionally, for the process and measurement noise covariances: $\Omega = Q_{\Omega}D_{\Omega}Q_{\Omega}^T$ and $R = Q_{R}D_{R}Q_{R}^T$ , respectively. Set the initial values as follows: $Q_{P_{1|0}} = Q_{\Pi_0}$, $D^{1/2}_{P_{1|0}} = D^{1/2}_{\Pi_0}$ and $\hat x_{1|0} = \bar x_0$.

\textsc{Measurement Update}  ($k=1, \ldots, N$). Build the pre-arrays from the filter quantities that are currently available and, then, apply the SVD factorizations in order to obtain the corresponding SVD factors of the updated filter quantities as follows:
\begin{align}
\underbrace{
\begin{bmatrix}
D^{1/2}_{R} Q^T_{R} \\
D^{1/2}_{P_{k|k-1}}Q^T_{P_{k|k-1}}H^T
\end{bmatrix}
}_{\rm Pre-array}
 =
\underbrace{
\mathfrak{W}_{MU}^{(1)}
\begin{bmatrix}
D_{R_{e,k}}^{1/2} \label{svd:1} \\
0
\end{bmatrix}
Q_{R_{e,k}}^T
}_{\rm Post-array \: SVD \:factors},  \\
\bar K_{k} = \left(Q_{P_{k|k-1}}D_{P_{k|k-1}}Q^T_{P_{k|k-1}}\right)H^TQ_{R_{e,k}},  \label{svd:K} \\
\underbrace{
\begin{bmatrix}
D_{P_{k|k-1}}^{1/2}Q_{P_{k|k-1}}^T\left(I - K_k H\right)^T \\
D^{1/2}_{R} Q^T_{R}  K_{k}^T
\end{bmatrix}
}_{\rm Pre-array}
 =
\underbrace{
\mathfrak{W}_{MU}^{(2)}
\begin{bmatrix}
D_{P_{k|k}}^{1/2} \\
0
\end{bmatrix}
Q_{P_{k|k}}^T
}_{\rm Post-array \: SVD \:factors} \label{svd:2}
\end{align}
where we denote $K_k = \bar K_k D^{-1}_{R_{e,k}}Q^T_{R_{e,k}}$. The matrices $\mathfrak{W}_{MU}^{(1)} \in {\mathbb R}^{(m+n)\times (m+n)}$, $Q_{R_{e,k}} \in {\mathbb R}^{m\times m}$ and $\mathfrak{W}_{MU}^{(2)} \in {\mathbb R}^{(n+m)\times (n+m)}$, $Q_{P_{k|k}} \in {\mathbb R}^{n\times n}$ are the orthogonal matrices of the corresponding SVD factorizations in~\eqref{svd:1}, \eqref{svd:2}. Next, $D_{R_{e,k}}^{1/2} \in {\mathbb R}^{m\times m}$ and $D_{P_{k|k}}^{1/2} \in {\mathbb R}^{n\times n}$ are diagonal matrices with square roots of the singular values of $R_{e,k}$ and $P_{k|k}$, respectively.

It can be easily seen that the required SVD factors of the innovation covariance $R_{e,k}$, i.e. $\{Q_{R_{e,k}}, D_{R_{e,k}}^{1/2}\}$, and {\it a posteriori} error covariance matrix $P_{k|k}$, i.e. $\{Q_{P_{k|k}}, D_{P_{k|k}}^{1/2}\}$, are directly read-off from the post-array factors in~\eqref{svd:1} and \eqref{svd:2}, respectively. Finally, find {\it a posteriori} estimate $\hat x_{k|k}$ through equations
\begin{equation}
\hat x_{k|k}  = \hat x_{k|k-1}\!+\bar K_{k}D^{-1}_{R_{e,k}}\bar e_k, \; \; \bar e_k  =  Q_{R_{e,k}}^T\!\!\left(z_k-H\hat x_{k|k-1}\right)\!. \label{svd:3}
\end{equation}

\textsc{Time Update} ($k=1, \ldots, N$).
 Build the pre-array and apply the SVD factorization to obtain {\it a priori} error covariance SVD factors $\{Q_{P_{k+1|k}}, D_{P_{k+1|k}}^{1/2}\}$ as follows:
\begin{align}
\underbrace{
\begin{bmatrix}
D^{1/2}_{P_{k|k}}Q^T_{P_{k|k}}F^T\\
D^{1/2}_{\Omega} Q^T_{\Omega} G^T
\end{bmatrix}
}_{\rm Pre-array}
=
\underbrace{
\mathfrak{W}_{TU}
\begin{bmatrix}
D_{P_{k+1|k}}^{1/2} \\
0
\end{bmatrix}
Q_{P_{k+1|k}}^T
}_{\rm Post-array \: SVD \:factors} \label{svd:5}
\end{align}
 and find {\it a priori} estimate $\hat x_{k+1|k}$ as follows:
\begin{align}
\hat x_{k+1|k} & = F\hat x_{k|k} + Bu_{k}. \label{svd:6}
\end{align}

The SVD-based KF implementation above is formulated in two-stage form. Meanwhile, following~\cite{ParkKailath1995}, the conventional KF~\eqref{kf:1}~--\eqref{kf:4} is expressed in the so-called ``condensed'' form. Nevertheless, these KF variants are algebraically equivalent. It is easy to prove if we take into account the SVD factorization $A=\mathfrak{W} \Sigma \mathfrak{V}^T$ and the properties of orthogonal matrices. Indeed, for each pre-array to be decomposed we have $A^TA = (\mathfrak{V} \Sigma \mathfrak{W}^T)(\mathfrak{W}\Sigma \mathfrak{V}^T) = \mathfrak{V} \Sigma^2 \mathfrak{V}^T$. Next, by comparing both sides of the obtained matrix equations, we come to the corresponding SVD-based KF formulas. The detailed derivation can be found in~\cite{KulikovaOnline}.

%In summary, the SVD-based KF algorithms recursively update the SVD factors $D_{P_{k|k}}$ and $Q_{P_{k|k}}$ of the error covariance matrix $P_{k|k}$ at each iteration step of the filter instead of the full matrix $P_{k|k}$.  Such computational scheme is not free from roundoff errors, however, it ensures the symmetric form of $P_{k|k}=Q_{P_{k|k}}D_{P_{k|k}}Q_{P_{k|k}}^T$ at any time instance and its positive semi-definiteness. As a result, the SVD-based KF implementations are expected to be inherently more stable (with respect to roundoff errors) than the conventional KF implementation.

\section{Filter sensitivity equations: SVD-based approach} \label{Sec:SVD:diff}

To begin constructing the ``differentiated'' SVD-based method for computing the filter sensitivities, we pay attention to the underlying SVD-based filter and remark that it is formulated in the so-called array form. This makes the modern KF algorithms better suited to parallel implementation and to very large scale integration (VLSI) implementation as mentioned in~\cite{ParkKailath1995}. Each iteration of the SVD-based filter examined has the following pattern: given a pre-array $A \in {\mathbb R}^{(k+s)\times s}$, compute the post-array SVD factors $\mathfrak{W} \in {\mathbb R}^{(k+s)\times (k+s)}$, $\Sigma \in {\mathbb R}^{(k+s)\times s}$ and $\mathfrak{V} \in {\mathbb R}^{s\times s}$ by means of the SVD factorization
\begin{align} \label{numSVD}
A & = \mathfrak{W}\:\Sigma \:\mathfrak{V}^T, \quad
\Sigma =
\begin{bmatrix}
S  \\
0
\end{bmatrix},  \quad S={\rm diag}\{ \sigma_1,\ldots,\sigma_s\}
\end{align}
where the matrix $A$ is of full column rank, i.e. ${\rm rank} \:A = s$; the $\mathfrak{W}$, $\mathfrak{V}$ are orthogonal matrices and $S$ is a diagonal matrix with singular values of the pre-array $A$.

The goal of our study is to develop the method that naturally extends formula~\eqref{numSVD} on the post-array SVD factors' derivative computation. More precisely, the computational procedure is expected to utilize the pre-array $A$ and its derivative $A'_{\theta}$ for reproducing the SVD post-arrays $\{\mathfrak{W}, \Sigma, \mathfrak{V}\}$ together with their derivatives $\{\mathfrak{W}'_{\theta}, \Sigma'_{\theta}, \mathfrak{V}'_{\theta}\}$. To achieve our goal, we prove the result presented below. We also bear in mind that the SVD post-array factor $\mathfrak{W}$ is of no interest in the presented SVD-based KF for performing the next step of the filter recursion and, hence, the quantity $\mathfrak{W}'_{\theta}$ is not required to be computed.

\begin{Lm}
\label{lemma:1} Consider the SVD factorization in~\eqref{numSVD}.
%, assume that $\sigma_1>\ldots >\sigma_s$
Let entries of the pre-array $A(\theta)$ be known differentiable functions of a scalar parameter $\theta$. We assume that $\sigma_i(\theta) \ne \sigma_j(\theta)$, $j\ne i$, for all $\theta$. Given the derivative of the pre-array, $A'_{\theta}$, the following formulas calculate the corresponding derivatives of the post-arrays:
\begin{align}
\Sigma'_{\theta} & =
\begin{bmatrix}
S'_{\theta}  \\
0
\end{bmatrix},  & S'_{\theta} & = {\rm diag} \left[\mathfrak{W}^T A'_{\theta} \mathfrak{V}\right]_{s\times s}, \label{lemma:d:1} \\
\mathfrak{V}'_{\theta} & = \mathfrak{V} \left[ \bar L_2^T-\bar L_2 \right] && \label{lemma:d:2}
\end{align}
where $ \left[\mathfrak{W}^T A'_{\theta} \mathfrak{V}\right]_{s\times s}$ denotes the main $s \times s$ block of the matrix product $\mathfrak{W}^T A'_{\theta} \mathfrak{V}$, and
$\bar L_2$ is a strictly lower triangular matrix, which entries are computed as follows:
\begin{align}\label{barL2}
(\bar l_2)_{ij} & =\displaystyle\frac{{\bar u}_{ji}\sigma_j+{\bar l}_{ij}\sigma_i}{\sigma_i^2-\sigma_j^2},
& i & = 2,\ldots, s, \; j=1,\ldots, i-1.
\end{align}
In equation above, the quantities ${\bar u}_{ji}$ and ${\bar l}_{ji}$ denote the entries of matrices $\bar L$ and $\bar U$, respectively. The $\bar L$, $\bar U$ are strictly lower and upper triangular parts of the matrix product $\left[\mathfrak{W}^T A'_{\theta} \mathfrak{V}\right]_{s\times s}$, respectively.
\end{Lm}

\begin{proof}
By differentiating~\eqref{numSVD} with respect to $\theta$, we obtain
\begin{equation}\label{numdiffSVD}
A'_\theta=\mathfrak{W}'_\theta\Sigma \mathfrak{V}^T+\mathfrak{W} {\Sigma}'_\theta \mathfrak{V}^T + \mathfrak{W} \Sigma \: (\mathfrak{V}^T)'_{\theta}.
\end{equation}

Having applied a right-multiplier $\mathfrak{V}$ and a left-multiplier $\mathfrak{W}^T$ to equation~\eqref{numdiffSVD}, we have
\begin{equation}\label{baseSVD1}
\mathfrak{W}^T A'_{\theta} \mathfrak{V} = \left[\mathfrak{W}^T \mathfrak{W}'_{\theta}\right]\Sigma +\Sigma'_{\theta}+\Sigma \left[(\mathfrak{V}^T)'_{\theta}\mathfrak{V} \right].
\end{equation}
In deriving the equation above we take into account the properties of any orthogonal matrix $Q$, i.e. $QQ^T=Q^TQ=I$.

It is also easy to show that for any orthogonal matrix $Q$ the product $Q'_{\theta}Q^T$ is a skew symmetric matrix. Indeed, by differentiating both sides of $QQ^T=I$ with respect to $\theta$, we get $Q'_{\theta}Q^T+Q\left(Q^T\right)'_{\theta}=0$, or in the equivalent form $Q'_{\theta}Q^T=-\left(Q'_{\theta}Q^T\right)^T$. The latter implies that the matrix $Q'_{\theta}Q^T$ is skew symmetric.

For the sake of simplicity we introduce the following notations: $\Upsilon= \mathfrak{W}^T \mathfrak{W}'_{\theta}$ and $\Lambda=\mathfrak{V}^T\mathfrak{V}'_{\theta}$. As discussed above, the matrices $\Upsilon \in {\mathbb R}^{(k+s)\times (k+s)}$ and $\Lambda \in {\mathbb R}^{s\times s}$ are skew symmetric, because $\mathfrak{W}$ and $\mathfrak{V}$ are orthogonal matrices. Hence, we have $\Lambda^T = -\Lambda$. Taking into account this fact, we obtain the following partitioning of the matrix form of equation~\eqref{baseSVD1}:
\[
\begin{bmatrix}
\left[\mathfrak{W}^T A'_{\theta} \mathfrak{V}\right]_{s\times s}\\
\left[\mathfrak{W}^T A'_{\theta} \mathfrak{V}\right]_{k\times s}
\end{bmatrix}
\! = \!
\begin{bmatrix}
[\Upsilon]_{s\times s} & [\Upsilon]_{s\times k}\\
[\Upsilon]_{k\times s} & [\Upsilon]_{k\times k}
\end{bmatrix}
\begin{bmatrix}
S\\
0
\end{bmatrix}
\!+\!
\begin{bmatrix}
S'_{\theta}\\
0
\end{bmatrix}
\!-\!
\begin{bmatrix}
S\\
0
\end{bmatrix}
\Lambda.
\]

From the equation above, we derive the formula for the main $s\times s$ block of the matrix product $\mathfrak{W}^T A'_{\theta} \mathfrak{V}$
\begin{equation}\label{baseSVD4:new}
 \left[\mathfrak{W}^T A'_{\theta} \mathfrak{V}\right]_{s\times s} =
[\Upsilon]_{s\times s}S + S'_{\theta} - S\Lambda.
\end{equation}
Hence, the diagonal matrix $S'_{\theta}$ obeys the equation
\begin{equation}\label{baseSVD4}
S'_{\theta} = \left[\mathfrak{W}^T A'_{\theta} \mathfrak{V}\right]_{s\times s} - [\Upsilon]_{s\times s}S + S\Lambda.
\end{equation}

Now, let us discuss formula~\eqref{baseSVD4} in details. Recall the matrices $\Upsilon$ and $\Lambda$ are skew symmetric matrices and, hence, their diagonal entries are equal to zero. The multiplication of any skew symmetric matrix by a diagonal matrix does not change the matrix structure, i.e. the diagonal entries of the matrix products $[\Upsilon]_{s\times s}S$ and $S\Lambda$ are equal to zero as well. Meanwhile, the matrix $\left[\mathfrak{W}^T A'_{\theta} \mathfrak{V}\right]_{s\times s}$ is a full matrix and contains a diagonal part. Hence, from equation~\eqref{baseSVD4} we conclude that diagonal matrix $S'_{\theta}$ is, in fact, a diagonal part of the main $s\times s$ block of the matrix product $\mathfrak{W}^T A'_{\theta} \mathfrak{V}$. This completes the proof of formulas in equation~\eqref{lemma:d:1}.

Finally, we need to compute $\mathfrak{V}'_{\theta}$ where $\Lambda=\mathfrak{V}^T \mathfrak{V}'_{\theta}$. Since $\mathfrak{V}$ is an orthogonal matrix, we obtain $\mathfrak{V}'_{\theta} = \mathfrak{V} \Lambda$. Next, any skew symmetric matrix can be presented as a difference of a strictly lower triangular matrix and its transpose. Hence, the skew symmetric  matrices $\left[\Upsilon\right]_{s\times s}$ and $\Lambda$ can be represented as follows:
\begin{align} \label{L:decomposition}
\left[\Upsilon\right]_{s\times s} & = \bar L_1^T- \bar L_1 & \Lambda & = \bar L_2^T- \bar L_2
\end{align}
where $\bar L_1$ and $\bar L_2$ are strictly lower triangular matrices.

Next, we split the matrix product $\left[\mathfrak{W}^T A'_{\theta} \mathfrak{V}\right]_{s\times s}$ into strictly lower triangular, diagonal and strictly upper triangular parts,  i.e. $\left[\mathfrak{W}^T A'_{\theta} \mathfrak{V}\right]_{s\times s} = \bar L + D + \bar U$. It was proved above that $S'_{\theta}=D$. Taking into account this fact, the substitution of both formulas in~\eqref{L:decomposition} into~\eqref{baseSVD4} yields
\begin{equation}\label{baseSVD4:eq1}
\underbrace{D}_{S'_{\theta}}
 = \underbrace{\bar L + D + \bar U}_{\left[\mathfrak{W}^T A'_{\theta} \mathfrak{V}\right]_{s\times s}} - \underbrace{\left[\bar L_1^T- \bar L_1\right]}_{[\Upsilon]_{s\times s}}S + S \underbrace{\left[\bar L_2^T- \bar L_2\right]}_{\Lambda}.
\end{equation}
Hence, we obtain
\begin{equation}\label{baseSVD6}
\bar L+\bar U
=
[{\bar L}_1^T-\bar L_1]S
-
S[{\bar L}_2^T-\bar L_2].
\end{equation}
In~\eqref{baseSVD6}, the $\bar L$, $\bar L_1$, $\bar L_2$ are strictly lower triangular matrices, the $\bar U$ is a strictly upper triangular matrix and $S$ is a diagonal. Hence, equation~\eqref{baseSVD6} implies
\[
\left\{
\begin{array}{lcl}
\bar U & = & {\bar L}_1^T S-S{\bar L}_2^T,\\
\bar L & = & -{\bar L}_1 S+S{\bar L}_2.\\
\end{array}
\right.
\]
It can be solved with respect to entries of $\bar L_2$ as follows:
\[
(\bar l_2)_{ij}=\displaystyle\frac{{\bar u}_{ji}s_{j}+{\bar l}_{ij}s_i}{s_i^2-s_j^2}, \; i=2, \ldots, s, \; j=1,\ldots,i-1.
\]

The formula above is exactly equation~\eqref{barL2}. Having computed the entries $(\bar l_2)_{ij}$ we can form the matrix $\Lambda = \bar L_2^T- \bar L_2$ in~\eqref{L:decomposition} and, then, compute   the derivative $\mathfrak{V}'_{\theta} = \mathfrak{V} \Lambda$. This completes the proof of~\eqref{lemma:d:2} and Lemma~1.
\end{proof}

\begin{remark}
The assumption of singular values of $A(\theta)$ being distinct for all values of parameter $\theta$ is necessary for avoiding the division by zero in formula~\eqref{barL2}. In future, if possible, we will intend for relaxing this restriction, which reduces the practical applicability of the proposed method.
\end{remark}

For readers' convenience, Algorithm~1 provides a pseudocode for the computational scheme derived in Lemma~1.
%_______________________________________
\begin{codebox}
\Procname{$\proc{Algorithm 1. Differentiated SVD}(A, A'_{\theta})$}
\zi {\bf Input:} $A$, $A'_{\theta}$ \qquad \Comment{\it \small Pre-array and its derivative}
\li Apply SVD from~\eqref{numSVD} to the pre-array $A$. Save $\mathfrak{W}$, $S$, $\mathfrak{V}$.
\li Compute the matrix product $\mathfrak{W}^T A'_{\theta} \mathfrak{V}$.
\li Extract the main $s\times s$ block $M=\left[\mathfrak{W}^T A'_{\theta} \mathfrak{V}\right]_{s\times s}$.
\li $M=\bar L + D + \bar U$. \Comment{\it \small  Split into strictly lower triangular, diagonal}
\zi \phantom{$M=\bar L + D + \bar U$.} \Comment{\it \small  and strictly upper triangular parts}
\li Given $\bar L$, $\bar U$ and $S$, compute the lower triangular $\bar L_2$ by~\eqref{barL2}.
\li Find $\mathfrak{V}'_{\theta} = \mathfrak{V} \left[ \bar L_2^T-\bar L_2 \right]$.
\li Find $S^{\prime}_{\theta} = D$. Hence, $\Sigma^{\prime}_{\theta} = \left[S^{\prime}_{\theta} \; | \;  0 \right]^T$.
\zi {\bf Output:} $\Sigma$, $\mathfrak{V}$ and $\Sigma'_{\theta}$, $\mathfrak{V}'_{\theta}$ \Comment{\it \small Post-arrays and their derivative}
\end{codebox}
%_______________________________________

The theoretical result presented in Lemma~1 can be further applied to the SVD factorization-based KF discussed in Section~\ref{filters}. The obtained computational scheme is summarized in Algorithm~2 and shown in the form of a pseudocode. The new ``differentiated'' SVD-based KF extends the underlying SVD-based filter on the derivative computation (with respect to unknown system parameters) for updating the corresponding filter sensitivities equations. The method can be used for replacing the conventional ``differentiated KF'' approach  discussed in Section~\ref{sec:state} by inherently more stable approach, which is preferable for practical implementation. Finally, we would like to remark that any ``differentiated'' filtering technique consists of two parts: i) the underlying KF variant, and ii) its ``differentiated'' extension used for the filter sensitivities computation.

\begin{codebox}
\Procname{$\proc{Algorithm~2. Differentiated SVD-based KF}(\bar x_0, \Pi_0)$}
\zi {\bf Initial Step}  ($k=0$).
\li \; $\Omega = Q_{\Omega}D_{\Omega}Q_{\Omega}^T$ and $R = Q_{R}D_{R}Q_{R}^T$ \Comment{\it \small SVD factorization}
\li \; $\Pi_0 = Q_{\Pi_0}D_{\Pi_0}Q_{\Pi_0}^T$ \qquad \qquad \qquad \quad  \Comment{\it \small SVD factorization}
\li \; Set $Q_{P_{1|0}} = Q_{\Pi_0}$, $D^{1/2}_{P_{1|0}} = D^{1/2}_{\Pi_0}$ and $\hat x_{1|0} = \bar x_0$.
\li \; Set $\Di{Q_{P_{1|0}}}\!\! = \Di{Q_{\Pi_0}}$, $\Di{D^{1/2}_{P_{1|0}}}\!\! = \Di{D^{1/2}_{\Pi_0}}$, $\Di{\hat x_{1|0}} = 0$.
\zi {\bf Measurement Update:}  ($k=1, \ldots, N$).
\li \; Build pre-array $A$ in~\eqref{svd:1} and its derivatives $\Di{A}$, $i=\overline{1,p}$.
\li \; $\left[ \Sigma, \: \mathfrak{V}, \:\Di{\Sigma}, \:\Di{\mathfrak{V}} \right] \: \leftarrow$ \verb"Differentiated SVD"($A$, $\Di{A}$).
\li \; $\left\{ D^{1/2}_{R_{e,k}}, \: \Di{D^{1/2}_{R_{e,k}}} \right\} \: \leftarrow$ read-off from $\Sigma$, $\Di{\Sigma}$ ($i=\overline{1,p}$).
\li \; $\left\{ Q_{R_{e,k}}, \: \Di{Q_{R_{e,k}}} \right\} \: \leftarrow$ read-off from $\mathfrak{V}$, $\Di{\mathfrak{V}}$ ($i=\overline{1,p}$).
\li \; Find $\bar K_{k}$ from~\eqref{svd:K} and $K_k = \bar K_k D^{-1}_{R_{e,k}}Q^T_{R_{e,k}}$.
\li \; $\Di{\bar K_{k}} = \Di{\left(Q_{P_{k|k-1}}D_{P_{k|k-1}}Q^T_{P_{k|k-1}}H^TQ_{R_{e,k}}\right)}$, $i=\overline{1,p}$.
%\zi \; $\phantom{\Di{\bar K_{k}}} + Q_{P_{k|k-1}}\DDi{D_{P_{k|k-1}}}Q^T_{P_{k|k-1}}H^TQ_{R_{e,k}}$
%\zi \; $\phantom{\Di{\bar K_{k}}} + Q_{P_{k|k-1}}D_{P_{k|k-1}}\DDi{Q^T_{P_{k|k-1}}}H^TQ_{R_{e,k}}$
%\zi \; $\phantom{\Di{\bar K_{k}}} + Q_{P_{k|k-1}}D_{P_{k|k-1}}Q^T_{P_{k|k-1}}\DDi{H}^TQ_{R_{e,k}}$
%\zi \; $\phantom{\Di{\bar K_{k}}} + Q_{P_{k|k-1}}D_{P_{k|k-1}}Q^T_{P_{k|k-1}}H^T\DDi{Q_{R_{e,k}}}$, $i=\overline{1,p}$.
\li \; Build pre-array $A$ in~\eqref{svd:2} and its derivatives $\Di{A}$, $i=\overline{1,p}$.
\li \; $\left[ \Sigma, \: \mathfrak{V}, \:\Di{\Sigma}, \:\Di{\mathfrak{V}} \right] \: \leftarrow$ \verb"Differentiated SVD"($A$, $\Di{A}$).
\li \; $\left\{ D^{1/2}_{P_{k|k}}, \: \Di{D^{1/2}_{P_{k|k}}} \right\} \: \leftarrow$ read-off from $\Sigma$, $\Di{\Sigma}$ ($i=\overline{1,p}$).
\li \; $\left\{ Q_{P_{k|k}}, \: \Di{Q_{P_{k|k}}} \right\} \: \leftarrow$ read-off from $\mathfrak{V}$, $\Di{\mathfrak{V}}$ ($i=\overline{1,p}$).
\li \; Find {\it a posteriori} estimate $\hat x_{k|k}$ and $\bar e_k$ from~\eqref{svd:3}.
\li \; $\Di{\bar e_k}  =  \DDi{Q_{R_{e,k}}^T}\left[z_k- H\hat x_{k|k-1}\right]$
\zi \; $\phantom{\Di{\bar e_k}}  - Q_{R_{e,k}}^T\left[\DDi{H}\hat x_{k|k-1}+H\DDi{\hat x_{k|k-1}}\right]$, \; $i=\overline{1,p}$.
\li \; $\Di{\hat x_{k|k}}  = \Di{\hat x_{k|k-1}} + \DDi{\bar K_{k}}D^{-1}_{R_{e,k}}\bar e_k +  \bar K_{k}D^{-1}_{R_{e,k}}\DDi{\bar e_k}$
\zi \; $\phantom{\Di{\hat x_{k|k}}} - \bar K_{k}D^{-1}_{R_{e,k}}\DDi{D_{R_{e,k}}}D^{-1}_{R_{e,k}}\bar e_k$, \; $i=\overline{1,p}$.
\zi {\bf Time Update:}  ($k=1, \ldots, N$).
\li \; Build pre-array $A$ in~\eqref{svd:5} and its derivatives $\Di{A}$, $i=\overline{1,p}$.
\li \; $\left[ \Sigma, \: \mathfrak{V}, \:\Di{\Sigma}, \:\Di{\mathfrak{V}} \right] \: \leftarrow$ \verb"Differentiated SVD"($A$, $\Di{A}$).
\li \; $\left\{ D^{1/2}_{P_{k+1|k}}, \: \Di{D^{1/2}_{P_{k+1|k}}} \right\} \: \leftarrow$ read-off from $\Sigma$, $\Di{\Sigma}$ ($i=\overline{1,p}$).
\li \; $\left\{ Q_{P_{k+1|k}}, \: \Di{Q_{P_{k+1|k}}} \right\} \: \leftarrow$ read-off from $\mathfrak{V}$, $\Di{\mathfrak{V}}$ ($i=\overline{1,p}$).
\li \; Find {\it a priori} estimate $\hat x_{k+1|k}$ from equation~\eqref{svd:6}.
\li \; $\Di{\hat x_{k+1|k}} = \DDi{F}\hat x_{k|k} + F\DDi{\hat x_{k|k}}+\DDi{B} u_{k}$, $i=\overline{1,p}$.
\zi {\bf End.}
\end{codebox}

At the same manner, one can naturally augment any existing SVD-based KF variant (see, for instance, the algorithms in~\cite{WangSVD1992,KulikovaOnline}) or  potentially new SVD-based KF implementation on the corresponding filter sensitivities computation.

Finally, taking into account the properties of orthogonal matrices, it is not difficult to show that the
negative log likelihood function (LF) given as~\cite{Schweppe1965}:
 \[
{\mathcal L}\left(\theta, Z_1^N\right)= c_0 +
\frac{1}{2} \sum \limits_{k=1}^N \left\{
 \ln\left(\det R_{e,k}\right)+ e_k^T R_{e,k}^{-1}e_k \right\}
\]
can be rewritten in terms of the SVD filter variables $Q_{R_{e,k}}$, $D_{R_{e,k}}$ and $\bar e_k$ appeared in equations~\eqref{svd:1}~--~\eqref{svd:6} as follows:
\begin{equation}
{\mathcal L}\left(\theta, Z_1^N\right) =
c_0 + \frac{1}{2} \sum \limits_{k=1}^N \left\{
  \ln\left(\det D_{R_{e,k}}\right)+ \bar e_k^T D^{-1}_{R_{e,k}}\bar e_k
\right\} \label{svd:LLF}
\end{equation}
where $Z_1^N=\{z_1,\ldots, z_N\}$ is $N$-step measurement history
and $c_0$ is a constant value where $c_0 = \frac{Nm}{2}\ln(2\pi)$.

Taking into account that the matrix $D_{R_{e,k}}$ is diagonal and using the Jacobi's formula, $\D{\left(A^{-1}\right)} = -A^{-1} \DD{A} A^{-1}$,
from~\eqref{svd:LLF} we obtain the expression for the log LF gradient evaluation in terms of the SVD filter variables and their derivatives computed in the newly-developed Algorithm~2 (for each $i=1, \ldots, p$):
\begin{align}
\Di{{\mathcal L}}\left(\theta, Z_1^N\right) & =
\frac{1}{2}\sum \limits_{k=1}^N \left\{
\tr { \DDi{D_{R_{e,k}}}D^{-1}_{R_{e,k}}} +2\DDi{ \bar e_k^T}D^{-1}_{R_{e,k}}\bar e_k
 \right. \nonumber \\
&  \left.  -\bar e_k^TD^{-1}_{R_{e,k}} \DDi{D_{R_{e,k}}} D^{-1}_{R_{e,k}}\bar e_k
\right\}. \label{diff:svd:LLF}
\end{align}

\section{Numerical experiments } \label{experiments}

By using simple test problem, we would like to demonstrate thoroughly each step of the method summarized in Algorithm~1.
\begin{exmp}\label{ex:1}
Given pre-array $A(\theta)$ and its derivative $A'_{\theta}$
\begin{equation*}\label{A:1}
 A(\theta) =
 \left[\begin{array}{rr}
-2\theta  & \sin(\theta) \\
2\theta  & \theta^2 \\
\sin^2{(\theta)}  &  1/3\:\theta^3\\
\theta  & 2\theta^2-1\\
\cos^2{(\theta)} & \theta^3+\theta^2
\end{array}\right],
%A'_{\theta}= \left[
%\begin{array}{rr}
%-2  & \cos(\theta) \\
%2  &  2\theta \\
%2\cos(\theta)\sin(\theta)  & {\theta^2}\\
%1  &  4\theta\\
%-2\sin(\theta)\cos(\theta)  &  3\theta^2+2\theta
%\end{array}\right]
\end{equation*}
compute the corresponding SVD post-arrays $\Sigma$, $\mathfrak{V}$ and their derivative ${\Sigma}'_{\theta}$, $\mathfrak{V}'_{\theta}$ at the point $\hat \theta=0.5$.
\end{exmp}

Table~1 illustrates each step of the computational scheme in Algorithm~1. To  assess the accuracy of computations, we compute
$l_{\infty} = \left|\left|{(A^T A)}'_{\hat\theta =0.5}-{(\mathfrak{V} \Sigma^2 \mathfrak{V}^T)}'_{\hat\theta=0.5}\right|\right|_{\infty}$. This quantity should be small. Indeed, taking into account the properties of diagonal and orthogonal matrices, from~\eqref{numSVD} we have $A^T A = \mathfrak{V} \Sigma^T \mathfrak{W}^T \mathfrak{W} \Sigma \mathfrak{V}^T = \mathfrak{V} \Sigma^2 \mathfrak{V}^T$. Hence, the derivatives of both sides of the last formula should coincide as well. In our numerical experiment we obtain $l_{\infty} = 1.99 \cdot 10^{-15}$. This justifies the correctness of computations via Algorithm~1 and  confirms the theoretical derivations in Lemma~1.

\begin{table}[h]
\caption{Algorithm~1 illustrative calculations for Example~1}  \label{tab:matrix}
\centering
\begin{tabular}{|l|l|}
\hline
{\bf  Input}  & Pre-array: $\left.A \right|_{\hat \theta = 0.5} = \left[
\begin{smallmatrix*}[r]
 -1.0000  &  0.4794\\
    1.0000  &  0.2500\\
    0.2298  &  0.0417\\
    0.5000 &  -0.5000\\
    0.7702 &   0.3750
\end{smallmatrix*}
\right]^{\phantom{M^M}}$ \\[10pt]
& Pre-array derivative: $\left.A'_{\theta}\right|_{\hat \theta = 0.5}=
\left[
\begin{smallmatrix*}[r]
 -2.0000    &  0.8776\\
    2.0000  &  1.0000\\
    0.8415  &  0.2500\\
    1.0000  &  2.0000\\
   -0.8415  &  1.7500
\end{smallmatrix*}
\right]$  \\[10pt]
\hline
Line 1. &
$\mathfrak{W}=\left[
\begin{smallmatrix*}[r]
-0.6070 &   0.4848  &  0.1556  &  0.2057  &  0.5745\\
    0.5723  &  0.4035  &  0.0539 &  -0.5533  &  0.4478\\
    0.1323  &  0.0735  &  0.9579  &  0.1059  & -0.2197\\
    0.3159  & -0.5593  &  0.0946  &  0.4337  &  0.6247\\
    0.4321  &  0.5329 &  -0.2152  &  0.6724  & -0.1756
\end{smallmatrix*}
\right]^{\phantom{M^M}}$ \\[15pt]
& $\Sigma=\left[
\begin{smallmatrix*}[r]
1.7061    &     0\\
         0  &  0.8185\\
         0     &    0\\
         0      &   0\\
         0     &    0
\end{smallmatrix*}
\right]$, $\mathfrak{V}=\left[
\begin{smallmatrix*}[r]
  0.9967  &  0.0811\\
   -0.0811  &  0.9967
\end{smallmatrix*}
\right]$ \\[10pt]
Line 2. & Compute $M=
\left[
\begin{smallmatrix*}[r]
    2.2959 &   -1.6522 \\
    1.1584 &   0.5691 \\
    0.5177 &  -0.1427 \\
   -0.2470 &  -2.2944 \\
   -1.8181 &  -0.8517
\end{smallmatrix*}
\right]$.  \\[15pt]
Line 3. & Extract $\left[M\right]_{2\times2}  =
\left[
\begin{smallmatrix*}[r]
    2.2959 &   -1.6522 \\
    1.1584 &   0.5691
\end{smallmatrix*}
\right]$ \\[10pt]
Line 4. &
Split $\left[M\right]_{2\times2} \! = \! \left[
\begin{smallmatrix}
 0  &  0\\
    1.1584   &  0
 \end{smallmatrix}
\right]\!+\!
\left[
\begin{smallmatrix}
  2.2959   &  0\\
    0  &  0.5691
\end{smallmatrix}
\right]
\!+\!
\left[
\begin{smallmatrix}
0  &    -1.6522\\
    0  &  0
\end{smallmatrix}
\right]$\!\!\!
\\[10pt]
Line 5. & Compute
$\bar L_2 =
\left[
\begin{smallmatrix}
      0 &  0\\
    0.8348  &     0
\end{smallmatrix}
\right]$\\[10pt]
Line 6. & $\left.\mathfrak{V}'_{\theta}\right|_{\hat \theta = 0.5} = \left[
\begin{smallmatrix*}[r]
0.0677  & -0.8321\\
    0.8321  &  0.0677
\end{smallmatrix*}
\right]$  \\[10pt]
Line 7. & $\left.\Sigma'_{\theta}\right|_{\hat \theta = 0.5} = \left[
\begin{smallmatrix*}[c]
 2.2959   &      0\\
         0  &  0.5691\\
         0    &     0\\
         0     &    0\\
         0      &   0
\end{smallmatrix*}
\right]$ \\[10pt]
\hline
{\bf Output}  &
 Post-arrays:
$\left.\Sigma\right|_{\hat \theta = 0.5} = \left[
\begin{smallmatrix*}[c]
1.7061       &  0 \\
         0    & 0.8185 \\
         0    &     0\\
         0     &    0\\
         0      &   0
\end{smallmatrix*}
\right]^{\phantom{M^M}}$ \\[10pt]
& \phantom{Post-arrays:} $\left. \mathfrak{V}\right|_{\hat \theta = 0.5} = \left[
\begin{smallmatrix*}[r]
    0.9967  & -0.0811 \\
   -0.0811  & -0.9967
\end{smallmatrix*}
\right]$ \\
&  Post-arrays' derivative: $\left.\Sigma'_{\theta}\right|_{\hat \theta = 0.5}$ and  $\left.\mathfrak{V}'_{\theta}\right|_{\hat \theta = 0.5}$ (Lines 6,7) \\
\hline
\end{tabular}
\end{table}

Next, we wish to demonstrate how the novel method for the filter sensitivities evaluation (Algorithm~2) works in practice. For that, we consider the parameter estimation problem where the gradient-based optimization method is applied for finding the optimal value of unknown system parameters. We test the conventional ``differentiated'' KF (Eqs~\eqref{kf:1}~-- \eqref{diff:kf:4} in Section~\ref{sec:state}) and the previously derived SR- and UD-based ``differentiated'' KF variants from~\cite{Kulikova2009IEEE} and~\cite{Tsyganova2013IEEE}, respectively, against the new ``differentiated'' SVD-based KF (Algorithm~2). As discussed in Section~\ref{Sec:SVD:diff}, all ``differentiated'' methods consist of two parts and, hence, they compute the Log LF and its gradient simultaneously. These values are utilized by a gradient-based optimization method for maximizing the log LF with respect to system parameters. Our library of codes is implemented in MATLAB where we use the built-in optimization method \verb"fminunc".

\begin{exmp}\label{ex:2}
Consider a linearized version of the in-track motion dynamic when a satellite travels in a circular orbit~\cite{Rauch1965}:
\begin{align*}
x_{k} & =
\begin{bmatrix}
1 & 1 & 0.5 &  0.5 \\
0 & 1 & 1 & 1 \\
0 & 0 & 1 & 0 \\
0 & 0 & 0 & 0.606
\end{bmatrix}
x_{k-1} + w_{k-1}, & \Omega & = \begin{bmatrix}
0 & 0 & 0 & 0 \\
0 & 0 & 0 & 0 \\
0 & 0 & 0 & 0 \\
0 & 0 & 0 & q_1
\end{bmatrix}  \\
z_k & =
\begin{bmatrix}
1 & 1 & 1 & 1\\
1 & 1 & 1 & 1+\delta
\end{bmatrix}
x_k + v_k, & R & = \begin{bmatrix}
\theta^2 \: \delta^2 & 0 \\
0 & \theta^2 \: \delta^2
\end{bmatrix}
\end{align*}
where $q_1 = 0.63 \cdot 10^{-2}$, $x_0 \sim {\mathcal N}(0, \theta^2 I_4)$ and $\theta$ is the unknown system parameter that needs to be estimated. In contrast to~\cite{Rauch1965}, we wish to test both well-conditioned and ill-conditioned situations. For that, following~\cite{GrewalAndrews2001}, we simulate the roundoff by parameter $\delta$. It is assumed to be
$\delta^2<\epsilon_{roundoff}$, but $\delta>\epsilon_{roundoff}$
where $\epsilon_{roundoff}$ denotes the unit roundoff
error\footnote{Computer roundoff for floating-point arithmetic is
characterized by a single parameter $\epsilon_{roundoff}$,
defined in different sources as the largest number such that either
$1+\epsilon_{roundoff} = 1$ or $1+\epsilon_{roundoff}/2 = 1$ in
machine precision. }. When $\delta \to \epsilon_{roundoff}$, i.e. the machine precision limit, the problem above becomes ill-conditioned. By varying the ill-conditioning parameter $\delta$, we are able to explore some numerical insights of each method assessed.
\end{exmp}

The numerical experiment is organized as follows. For each fixed value of ill-conditioning parameter $\delta$, the SSM in Example~2 is simulated for $\theta^* = 5$ to generate $N=100$ measurements. Next, the unknown system parameter $\theta$ is estimated from the available experimental data, $Z_1^N = \{z_1, \ldots, z_N \}$, by using gradient-based adaptive KF techniques examined, i.e. by four ``differentiated'' KF methods mentioned earlier in this Section. For a fair comparison, each ``differentiated'' algorithm utilizes the same data $Z_1^N$ and the same initial value for the optimization method, $\hat \theta^{(0)} = 1$. Next, the obtained optimal estimate $\hat \theta^*$ is compared with the ``true'' value of $\theta^* = 5$ for assessing the estimation quality of each method. We repeat the experiment $M=100$ times and calculate {\it a posterior} mean of the estimate, the root mean squared error (RMSE) and the mean absolute percentage error (MAPE) over $100$ Monte Carlo runs.

\begin{table*}
\renewcommand{\arraystretch}{1.3}
\caption{Effect of roundoff errors in ill-conditioned test problems in Example~2; exact $\theta^* =5$, $100$ Monte Carlo runs} \label{MC-estimators}
\centering
\begin{tabular}{c||c|c|c||c|c|c||c|c|c||c|c|c}
\hline
&  \multicolumn{3}{c||}{``differentiated" KF} &
\multicolumn{3}{c||}{``differentiated" SR-based KF} & \multicolumn{3}{c||}{``differentiated" UD-based KF} & \multicolumn{3}{c}{``differentiated" SVD-based KF}\\
\cline{2-13}
$\delta$ &  Mean &  RMSE & MAPE\% &  Mean & RMSE & MAPE\% &  Mean & RMSE & MAPE\% &  Mean & RMSE & MAPE\%  \\
\hline
 $ 10^{-1\phantom{0}}$ &  5.0046 &   0.2485  &  3.8829  & 5.0046 &   0.2485  &  3.8829 & 5.0046 &   0.2485  &  3.8829 & 5.0046 &   0.2485  &  3.8829\\
 $ 10^{-2\phantom{0}}$ &  4.9649 &   0.2784 &   4.2892  & 4.9649 &   0.2784 &   4.2883 & 4.9649  &  0.2784 &   4.2883 & 4.9649  &  0.2784 &   4.2883\\
 $ 10^{-3\phantom{0}}$ &  5.2764 &   0.7027 &   9.7757  & 5.0083 &   0.3555 &   5.7217 & 5.0083  &  0.3555 &   5.7217 & 5.0083  &  0.3555 &   5.7217 \\
 $ 10^{-4\phantom{0}}$ &  8.8812 &   4.1440 &   77.623  & 4.9879 &   0.3715 &   5.8595 & 4.9879  &  0.3715 &   5.8596 & 4.9879  &  0.3715 &   5.8597\\
 $ 10^{-5\phantom{0}}$ &  0.2803 &   8.0217 &   $>$100\%& 4.9509 &   0.3352 &   5.6154 & 4.9508  &  0.3353 &   5.6162 & 4.9509  &  0.3352 &   5.6150\\
 $ 10^{-6\phantom{0}}$ & -0.1315 &   7.2403 &   $>$100\%& 4.9310 &   1.0362 &   6.8368 & 4.9323  &  1.0333 &   6.8265 & 5.0288  &  0.3138 &   4.8826\\
 $ 10^{-7\phantom{0}}$ & $-$ & $-$ & $-$ & 4.9298  &  0.3658  &  5.8586 & 4.9268  &  0.3562 &   5.6883 & 4.9249 &   0.3507 & 5.5674 \\
 $ 10^{-8\phantom{0}}$ & $-$ & $-$ & $-$ & $-$ & $-$ & $-$ & 5.0437 &   0.3757  &  6.0712   &   5.0493 & 0.3790 &   6.0946 \\
 $ 10^{-9\phantom{0}}$ & $-$ & $-$ & $-$ & $-$ & $-$ & $-$ & 6.0119 &   1.2179  &  20.762   &   5.9738 & 1.1853 &   20.106 \\
$ 10^{-10\phantom{0}}$ & $-$ & $-$ & $-$ & $-$ & $-$ & $-$ & 6.7496 &   2.6030  &  49.405   &   6.7021 & 2.5286 &   49.252 \\
  \hline
\end{tabular}
\end{table*}

Having carefully analyzed the obtained numerical results summarized in Table~2, we make a few important conclusions.
First, all ``differentiated'' KF variants work equally well when $\delta$ is about $10^{-1}$ and $10^{-2}$, i.e. when the problem is not ill-conditioned. This  confirms that all ``differentiated'' techniques are algebraically equivalent. Second, among all methods examined, the conventional approach (``differentiated'' KF) shows the worst performance. It degrades faster than any other algorithms when $\delta \to \epsilon_{roundoff}$. Furthermore, the line in Table~2 means that MATLAB can not even run the algorithm. Third, we analyze the outcomes obtained by other methods tested and observe that the UD- and SVD-based ``differentiated'' techniques produce a better estimation quality than the SR-based counterpart. This conclusion is reasonable if we recall that in this paper we do not explore the filtering algorithms, but their differential form for the KF sensitivities computation. Any existing ``differentiated'' SR-based scheme requires the triangular matrix inversion $R_{e,k}^{1/2}$ that is a square-root factor of the innovation covariance $R_{e,k}$; see Eq~(6) in~\cite{Kulikova2009IEEE}. In contrast, the UD- and SVD-based ``differentiated'' methods involve the inversion of only diagonal matrix $D_{R_{e,k}}$; see~\eqref{diff:svd:LLF} and Eq~(8) in~\cite{Tsyganova2013IEEE}. Finally, we observe that the new SVD-based approach slightly outperforms the UD-based counterpart  when $\delta \to \epsilon_{roundoff}$.

In summary, the previously derived UD- and the new SVD-based techniques provide the best estimation quality when solving parameter estimation problem by the gradient-based adaptive filtering methodology. This creates a strong background for their practical use. In our ill-conditioned test example, the new SVD-based approach even slightly outperforms the UD-based counterpart.

%\bibliographystyle{IEEEtran}
%\bibliography{BibTex_Library/books,%
%              BibTex_Library/list_MVKulikova,%
%              BibTex_Library/list_Tsyganova,%
%              BibTex_Library/list_identification,%
%              BibTex_Library/filters,%
%              BibTex_Library/list_ML,%
%              BibTex_Library/SVD,%
%              BibTex_Library/list_linalg}

\end{document}